\documentclass[12pt]{report} 

\usepackage[a4paper,portrait,total={17cm,25cm}]{geometry}

\usepackage{authblk}
\title{Non-recursive trade-offs for two-way machines\\
{\small\textit{International Journal of Foundations of Computer Science}, 16:943–956, 2005.}}
\author{Christos Kapoutsis}
\affil{Computer Science and Artificial Intelligence Laboratory\\
    Massachusetts Institute of Technology\\
    Cambridge, MA 02139, USA, \texttt{cak@mit.edu}}
\date{}

\usepackage[centertags]{amsmath}
\usepackage{amsthm}
\newcommand{\modelofcomputation}[1]{\text{\scshape\mdseries\rmfamily #1}}

\newcommand{\odfa}{\modelofcomputation{{{\footnotesize 1}dfa}}}
\newcommand{\onfa}{\modelofcomputation{{{\footnotesize 1}nfa}}}

\newcommand{\dca}{\modelofcomputation{{dca}}}

\newcommand{\odpa}{\modelofcomputation{{\footnotesize 1}dpa}}
\newcommand{\onpa}{\modelofcomputation{{\footnotesize 1}npa}}

\newcommand{\tm}{\modelofcomputation{tm}}
\newcommand{\dtm}{\modelofcomputation{dtm}}

\newcommand{\complexityproblem}[1]{\text{\upshape\mdseries\sffamily #1}}

\newcommand{\looping}{\overline{\complexityproblem{HALTING}}}

\newcommand{\ggG}{\varGamma}

\newcommand{\ggP}{\varPi}
\newcommand{\ggS}{\varSigma}

\newcommand{\set}[1]{\{#1\}}

\newcommand{\amp}{\;\;\&\;\;}

\newcommand{\lang}[1]{{(#1)}}
\newcommand{\desc}[1]{{\langle#1\rangle}}

\newcommand{\s}{^*}

\newcommand{\nats}{{\mathbb{N}}}

\newcommand{\sy}[1]{\textup{\texttt{#1}}}

\newcommand{\bla}{\sqcup}
\newcommand{\dotz}{\dot{\hspace{-.3ex}\sy{0}}}
\newcommand{\doto}{\dot{\hspace{-.2ex}\sy{1}}}
\newcommand{\symz}{\sy{0}}
\newcommand{\symo}{\sy{1}}

\newcommand{\FF}{\complexityproblem{f}}

\newcommand{\PP}{\ggP}
\newcommand{\PPy}{\ggP_{\textup{yes}}}
\newcommand{\PPn}{\ggP_{\textup{no}}}
\newcommand{\II}{\complexityproblem{I}}
\newcommand{\IIy}{\II_{\textup{yes}}}
\newcommand{\IIn}{\II_{\textup{no}}}
\newcommand{\eot}{\complexityproblem{E}}
\newcommand{\eoty}{\eot_{\textup{yes}}}
\newcommand{\eotn}{\eot_{\textup{no}}}

\newcommand{\LLL}[1]{(#1\/)}
\newcommand{\SSS}[1]{|#1\/|}

\newenvironment{nicelist}
{\begin{list}
   {$\bullet$}
   {\setlength{\topsep}{0pt}
    \setlength{\parsep}{0pt}
    \setlength{\itemsep}{0pt}
    \setlength{\labelwidth}{\leftmargin}
    \setlength{\itemindent}{0ex}
    \setlength{\listparindent}{0pt}}}
{\end{list}} 
\usepackage{rotating}

\theoremstyle{plain}
\newtheorem{theorem}{Theorem}
\newtheorem{lemma}[theorem]{Lemma}

\newtheorem*{theorem*}{Theorem}
\newtheorem*{maintheorem*}{Main theorem}
\newtheorem*{mainlemma*}{Main lemma}
\newtheorem*{lemma*}{Lemma}

\theoremstyle{definition}

\newtheorem*{definition*}{Definition}

\newtheorem{xample}[theorem]{Example}

\usepackage{amssymb}

\begin{document}
\maketitle

\begin{abstract}
If the machines of some type~$\modelofcomputation{a}$ have enough
resources to (i)~solve problems that no machine of
type~$\modelofcomputation{b}$ can solve, and (ii)~simulate any unary
two-way deterministic finite automaton that has access to a
linearly-bounded counter, then typically no recursive function can
upper bound the increase in the size of description when a machine of
type~$\modelofcomputation{a}$ is replaced by an equivalent machine of
type~$\modelofcomputation{b}$.  
\end{abstract}

\section{Introduction}
\label{sec:introduction}

The origins of descriptional complexity of formal machines can be
traced back to~\cite{rasc59}, when Rabin and Scott showed that every
one-way nondeterministic finite automaton ($\onfa$) has a
deterministic equivalent ($\odfa$) that is at most
\emph{exponentially} larger. The question whether this exponential
blowup in size can be avoided has been one of the first steps in the
general direction of investigating the relative succinctness of
different types of machines as language descriptors.

A similar but qualitatively distinct situation emerged several
years later, when Stearns~\cite{st67} studied the relation between
one-way deterministic pushdown automata ($\odpa$s) and $\odfa$s.
Differently from above, he knew that not every $\odpa$ has a
$\odfa$-equivalent. But, similarly as above, he proved that, whenever
such equivalents exist, the smallest among them are at most
\emph{triple-exponentially} larger than the~$\odpa$.

This naturally lead to the corresponding question for one-way
nondeterministic pushdown automata ($\onpa$s): in the case where a
$\onpa$ has $\odfa$-equivalents, what is an upper bound for the size
of the smallest among them? The answer was again qualitatively new, by
Meyer and Fischer~\cite{mefi71}, who showed that every such bound
grows (as a function of the size of the $\onpa$) faster than any
computable function. Put another way, among the cases where it is
possible to convert a $\onpa$ into a $\odfa$, the trade-off in the
size of description is in general \emph{non-recursive}. Several
refinements of this result followed~\cite{va76,scsz77}.

In an important development, Hartmanis~\cite{ha80} later explained
that the recursiveness of the trade-off from a type of
machines~$\modelofcomputation{a}$ to a not-as-powerful type of
machines~$\modelofcomputation{b}$ typically implies the
recognizability (semi-decidability) of the corresponding
\emph{inadequacy problem}: ``given a machine of
type~$\modelofcomputation{a}$, check that it has no
$\modelofcomputation{b}$-equivalents''. This greatly simplified the
proofs of~\cite{mefi71,va76,scsz77}, while it nicely revealed the
connections of the entire discussion to G\"odel's theorem that the
addition of an extra axiom to a formal system typically results in
non-recursively shorter proofs for some of its theorems~\cite{ha83}.

Today many refinements of the above results are known and
non-recursive trade-offs have emerged in numerous other comparisons
between different types of machines. Comprehensive surveys can be
found in~\cite{goea02,ku04}.

\subsection{This study} 

In a remark in~\cite{haba79a}, Hartmanis and Baker showed that \emph{a
non-recursive trade-off can occur even when an optimal algorithm
replaces a near-optimal one}.\footnote{The reader is referred
to~\cite{haba79a,ha83,ha89} for a quite interesting discussion of the
implications that this might have to our search for optimal
algorithms.} For example, converting an $n^{2+\epsilon}$-space
deterministic Turing machine ($\dtm$) into one that uses only
$n^2$-space involves a non-recursive blowup in the size of
description. In the pattern of~\cite{ha80}, they derived this
observation from the unrecognizability of the inadequacy problem from
near-optimal to optimal machines (from $n^{2+\epsilon}$-space to
$n^2$-space $\dtm$s), which in turn was shown to be a consequence of
the fact that the near-optimal complexity class is strictly larger
than the optimal one (some $n^{2+\epsilon}$-space $\dtm$s have no
$n^2$-space equivalent).

In the present study we refine that argument. We prove a general
theorem that directly shows the non-recursiveness of the trade-off in
numerous conversions between machines of different power. In loose
terms, our theorem states the following:

\medskip\noindent{\itshape If two types of machines
$\modelofcomputation{a}$~and~$\modelofcomputation{b}$ are such that
\begin{nicelist}
\item[\textup{(1)}] some machine of type~$\modelofcomputation{a}$ has
no equivalent machine of type~$\modelofcomputation{b}$, and
\item[\textup{(2)}] a machine of type~$\modelofcomputation{a}$ has
enough resources to simulate a unary two-way  deterministic finite
automaton which has access to a linearly-bounded counter,
\end{nicelist}
then the trade-off from machines of type~$\modelofcomputation{a}$ to
machines of type~$\modelofcomputation{b}$ is non-recursive.}

\medskip\noindent For example, for the previous remark on space, we
argue that, since $n^2 = o(n^{2+\epsilon})$, there exist
$n^{2+\epsilon}$-space $\dtm$s with no $n^2$-space equivalent, so that
condition~(1) is clearly true; that (2)~is also true follows from the
easy observation that any $\Omega(\lg n)$ amount of space suffices for
the simulation of a linearly-bounded counter.

The most characteristic applications concern the successive levels of
hierarchies of two-way multi\emph{pointer}\footnote{By `pointer' we
mean any of the following: a linearly-bounded counter, a \emph{blind}
read-only head (i.e., a head that cannot distinguish among 
different input symbols, although it \emph{can} distinguish between an
input symbol and an end-marker), an ordinary read-only head, a
\emph{sensing} read-only head (i.e., one that can sense which of
the other heads are at the same cell as itself), or a
pebble.\label{note:1}} automata. For example, the following trade-offs
are non-recursive:\footnote{In each case, the reference indicates
where Condition~(1) of the theorem has been established. For
Condition~(2), it is always easy to see that it is also satisfied.}
\begin{nicelist}
\item
from $k+1$ to $k$ counters, on linearly-bounded two-way deterministic
counter automata  (unary or not)~\cite{mo80},
\item
from $k+1$ to $k$ heads,  on two-way multi-head finite automata
(deterministic or not, unary or not)~\cite{mo76,mo77,mo80},
\item
from $k+1$ to $k$ heads,  on two-way multi-head pushdown automata
(deterministic or nondeterministic)~\cite{ib73},
\end{nicelist}
for all $k$. Sometimes, we can only be as refined as the hierarchy is
known to be:
\begin{nicelist}
\item
from $k+2$ to $k$ registers,  on linearly-bounded register machines
(deterministic or nondeterministic)~\cite{mo80},
\item
from $k+2$ to $k$ counters,  on linearly-bounded two-way
nondeterministic coun\-ter automata  (unary or not)~\cite{mo80},
\end{nicelist}
for all $k$. Similarly, the trade-off is non-recursive
\begin{nicelist}
\item
from~3 to~2 heads,   on a \emph{simple}\footnote{In a \emph{simple}
multi-head automaton, every input head after the first one is
\emph{blind} (cf. Footnote~\ref{note:1}).}  two-way deterministic
finite automaton~\cite{duga82}; it remains non-recursive even when we
start from a 2-head two-way deterministic finite automaton, or from a
1-head two-way deterministic pushdown automaton~\cite{duga82}.
\end{nicelist}
Finally, we can even conclude the non-recursiveness of the trade-off,
for all $k\geq 2$:
\begin{nicelist}
\item
from $k+1$ to $k$ work-tape symbols,   on $\tm$s (deterministic or
not) that, on every input of length $n$,  use no more than $\lg n$
work-tape cells~\cite{se77}  (even if the starting $\tm$ has unary
input alphabet, but  then only for sufficiently large~$k$).
\end{nicelist}
Other conversions between machines of different power can be
treated similarly.


Returning to the statement of the theorem above, we warn that it is,
in fact, incomplete. Additional conditions have to be met, concerning
$\modelofcomputation{a}$~and~$\modelofcomputation{b}$, their
descriptions, and how `size' is measured. However, in most interesting
cases these conditions are trivially satisfied (in the above examples
they are), so that listing them here would be a distraction. The
complete list is contained in the formal statement of the theorem in
Section~\ref{sec:main-theorem}.

\subsection{Outline} 
The next section describes the framework of this study in more
detail. Section~\ref{sec:main-theorem} states and proves the theorem,
except for an important lemma, which is proved in
Section~\ref{sec:main-lemma}, after some preparation in
Section~\ref{sec:programming-counters}. The discussion is abstract
enough to also cover cases where
$\modelofcomputation{a}$~or~$\modelofcomputation{b}$ denote types of
language descriptors other than machines (e.g., regular expressions,
grammars). For a more concrete discussion, see~\cite{ka04b}. (There, a
special case of the theorem is proved, namely the case concerning the
trade-off from $k+1$ to $k$ heads on two-way multi-head finite
automata.)

\section{Preliminaries}
\label{sec:framework}

We denote the set of positive integers by $\nats$. For $n,a\in\nats$,
we write $\lg_a n$ to denote $\lfloor\log_a n\rfloor$. The set of all
finite strings over an alphabet $\ggG$ is denoted by $\ggG\s$.

A (\emph{promise}) \emph{problem} over $\ggG$ is any pair
$\PP=(\PPy,\PPn)$ where $\PPy$, $\PPn$ are disjoint subsets of
$\ggG\s$. A $\tm$ \emph{recognizes} $\PP$ if it accepts all $x\in\PPy$
and rejects (possibly by looping) all $x\in\PPn$. If some $\tm$
recognizes $\PP$, we say $\PP$ is \emph{recursively enumerable} or
(\emph{Turing-}) \emph{recognizable}. For $\PP'$ also a problem over
$\ggG$, we write $\ggP\leq\ggP'$ and say $\PP$ \emph{reduces to}
$\PP'$ iff there is a $\tm$ that, on input $x\in\PPy\cup\PPn$,
eventually halts with an output $y$ such that  $ x\in\PPy \implies
y\in\PPy' $ and $ x\in\PPn \implies y\in\PPn' $.  If some
unrecognizable $\PP$ reduces to $\PP'$, then clearly $\PP'$ is
unrecognizable, as well.

If $\PPn=\overline{\PPy}$, then $\PP$ is also called a \emph{language}
and is adequately described by $\PPy$ alone. If in addition $\PPy$
contains exactly all sufficiently long strings for some interpretation
$0\leq l\leq\infty$ of `sufficiently long', $ \PPy = \set{ x\in\ggG\s
\mid \text{length($x$)} \geq l } $, then we say $\PP$ \emph{obeys a
threshold} (note that then $\PPy$ is empty iff this threshold is
infinite); a machine that solves $\PP$ is similarly said to obey the
same threshold.

\subsection{Descriptional systems}
\label{sec:descriptional-systems}

A \emph{descriptional system} over the alphabets $\ggS$ and $\ggG$ is
any set $D\subseteq\ggS\s$ of \emph{names} (or \emph{descriptors}),
along with two total functions $\LLL{\mbox{$\cdot$}}_D$ and
$\SSS{\mbox{$\cdot$}}_D$, mapping every name $d\in D$ to its
\emph{language} $\LLL{d}_D\subseteq\ggG\s$ and its \emph{size}
$\SSS{d}_D\in\nats$, respectively.

For a standard example, fix a binary encoding of all $\odfa$s with
input alphabet $\{\texttt{a},\texttt{b},\texttt{c}\}$, say. This
induces the descriptional system over $\{0,1\}$ and
$\{\texttt{a},\texttt{b},\texttt{c}\}$ that contains all encoding
strings as names and maps each of them to the language accepted by the
corresponding $\odfa$ (as its language) and to the number of states in
that $\odfa$ (as its size). (Alternatively, the size of a name could
just be its length.)

A system $D$ is \emph{decidable} if the membership problem for its
names is decidable. That is, if there exists a $\tm$ $U_D$ that always
halts and is such that:
$$ \text{for all $d\in D$ and $w\in\ggG\s$:} \qquad U_D(d,w)\text{
accepts}\iff w\in\LLL{d}_D.
$$ Thus, the system of the previous example is clearly decidable,
whereas a system containing binary encodings of $\tm$s would be
undecidable.

In order to be able to compare two descriptional systems $D$ and $E$
in terms of their relative succinctness, we require that they are
\emph{comparable}, in the sense that (i)~they are defined over the
same alphabets, and that (ii)~their $\LLL{\mbox{$\cdot$}}$ and
$\SSS{\mbox{$\cdot$}}$  mappings agree on all common
names,\footnote{It is only for simplicity that we require the
agreement for $\SSS{\mbox{$\cdot$}}$; we do not actually need it.}
$$ \text{for all $z\in D\cap E$: } \quad  \LLL{z}_D = \LLL{z}_E
\quad\text{ and }\quad \SSS{z}_D = \SSS{z}_E,
$$ so that subscripts can be dropped: for all $z\in D\cup E$,
$\LLL{z}$ and $\SSS{z}$ are unambiguous.  For such systems, the
comparison of $E$ against $D$ involves two natural notions:

\mbox{\textsc{i}. }For a name $e\in E$, there \emph{may} or \emph{may
not} exist a name in $D$ that maps to the same language. In the latter
case, we say that \emph{$D$ is inadequate for describing the language
of $e$} and, accordingly, we call the associated computational
problem, ``given an $e\in E$, check that no $d\in D$ maps to
$\LLL{e}$'', the \emph{inadequacy problem from $E$ to $D$}. Formally,
this is the promise problem $\II=(\IIy,\IIn)$, with:
\begin{align*}
\IIy &= \{ e\in E \mid (d)\neq(e)\text{, for all $d\in D$} \}, \\ \IIn
&= \{ e\in E \mid (d)=(e)\text{, for some $d\in D$} \}.
\end{align*}
Notice that $e$ is promised to be in $E$, so that solving $\II$ does
not require checking membership in $E$ (which might be hard, even
impossible).

\mbox{\textsc{ii}. }When a name $e\in E$ does have \emph{equivalent}
names in $D$ (i.e., names mapping to $\LLL{e}$), we naturally ask how
larger than $e$ the smallest of these $D$-equivalents are. The typical
way to answer this question is with a function $\FF:\nats\to\nats$
that upper bounds this increase in size, in terms of the size of
$e$. Namely, $\FF$ is such that
\begin{quote}
for all $s\in\nats$ and for all $e\in E$ of size $s$: if $D$ contains
names that are equivalent to $e$, then at least one of these names is
of size at most~$\FF(s)$.
\end{quote}
We say that $\FF$ \emph{upper bounds the trade-off} (for the
conversion) \emph{from $E$ to $D$}. When a computable such upper bound
exists, we say \emph{the trade-off from $E$ to $D$ is recursive}.

As first noted by Hartmanis~\cite{ha80}, discussions (\textsc{i}),
(\textsc{ii}) are not unrelated: unrecognizability of the inadequacy
problem typically implies the trade-off is non-recursive.

\begin{lemma}[Hartmanis] \label{lem:Hartmanis}
Suppose $D$, $E$ are two \emph{comparable} descriptional systems over
alphabets $\ggS$ and $\ggG$, and that the following conditions are met:
\begin{nicelist}
\item[\textup{(H$_1$)}]   both $D$ and $E$ are decidable,
\item[\textup{(H$_2$)}]   for every $e\in E$, we can effectively
compute its size $\SSS{e}$, and
\item[\textup{(H$_3$)}]   \label{item:H2i} there is a halting $\tm$
that, on input $s\in\nats$, produces a list $Z\subseteq\ggS\s$ such
that
\begin{nicelist}
\item[\textup{(i)}] the non-$D$ names can be recognized in $Z$:
$(Z\cap\overline{D},Z\cap D)$ is recognizable.
\item[\textup{(ii)}]  the languages of the $D$-names in $Z$ cover all
and only those languages over $\ggG$ that are supported by a name in
$D$ of size at most $s$:
$$ \{ \LLL{z} \mid z\in Z\cap D \} = \{ \LLL{d} \mid d\in D \amp
\SSS{d}\leq s \}.
$$
\end{nicelist}
\end{nicelist}
Then, recursiveness of the trade-off from $E$ to $D$ implies the
corresponding inadequacy problem is recognizable.
\end{lemma}

\noindent
Before giving the proof, let us remark how mild conditions
(H$_1$)--(H$_3$) are. For most interesting cases, the first two of
them are trivially true and (H$_3$) is satisfied via the $\tm$ that
simply lists all names in $D$ that have size $\leq s$ (so that the
problem of~(i) is trivially decidable and the two sets of~(ii)
trivially identical). Having (H$_3$) as complicated simply covers some
special cases (e.g., comparing general to unambiguous context-free
grammars~\cite[Example~2]{ha83}).

\proof{Suppose $D$, $E$ are as in the statement and $\FF$ is a
computable upper bound for the trade-off from $E$ to $D$.  To check
that a given $e\in E$ has no $D$-equivalents, we first compute
$s=\FF(\SSS{e})$ (by (H$_2$) and since $\FF$ is computable) and then
run the $\tm$ guaranteed by (H$_3$) on~$s$, to produce a (finite,
since the $\tm$ is halting) list of names $Z=\{z_1, z_2, \dots,
z_k\}$. At this moment, we know (by the selection of $\FF$ and
(H$_3$ii)) we should accept iff every $D$-name in $Z$ maps to a
language different from~$\LLL{e}$.  Equivalently, we should accept
iff: for every $z\in Z$,  \emph{either} $z$ is not a $D$-name
\emph{or} $z$ is a $D$-name and $\LLL{z}$, $\LLL{e}$ differ at one or
more $w\in\ggG\s$.

In order to check this, we start simulating, in two parallel threads:
[\textsc{i}]~the recognizer guaranteed by~(H$_3$i) on each of
$z_1,z_2,\dots,z_k$ in parallel, and [\textsc{ii}]~for all
$w\in\ggG\s$: the machines $U_E$ and $U_D$ (guaranteed by~(H$_1$))
respectively on $(e,w)$ and on each of
$(z_1,w),(z_2,w),\dots,(z_k,w)$. Whenever a $z\in Z$ is accepted in
thread~\textsc{i}, we cross it off the list; whenever a $z\in Z$ is
found to disagree with $e$ on some $w$ in thread~\textsc{ii}, it is
crossed off the list, as well; if the list ever gets empty, we accept.

Clearly, every string in $Z$ that is not a $D$-name, will eventually
be crossed off, in thread~\textsc{i}; similarly, each $D$-name that is
inequivalent to $e$ will also be eventually removed, in
thread~\textsc{ii}; moreover, neither thread can delete a $D$-name
that is equivalent to $e$. Hence, the list will eventually get empty
iff $e$ had no $D$-equivalent in the original list $Z$; which is true
iff $e$ has no $D$-equivalent at all.}

\subsection{Multi-counter automata}
\label{sec:multicounter-automata}

A \emph{deterministic automaton with $k$ counters}\footnote{In this
definition the reader will recognize \emph{the unary version} of the
\emph{two-way $(k-1)$-counter machines} of~\cite{fimero68,mo80} or of
other automata from elsewhere; they all describe the natural notion of
a \emph{unary two-way deterministic finite automaton that has
additional access to $k-1$ counters}. The present model is a bit
simpler, in the sense that the input is directly taken to be the upper
bound for one of the counters, saving us the redundant (in the unary
case) notion of an input tape. (Note that this definition differs from
that of~\cite{ka04b}, where the upper bound is applied to all
counters.)} ($\dca_k$) consists of a finite state control and $k$
counters, each of which can store a nonnegative integer. One of the
counters is distinguished as \emph{primary}, the rest being referred
to as \emph{secondary}. The input to the automaton is a nonnegative
upper bound $n$ for the primary counter. The machine starts at a
designated \emph{start state} with all its counters set to~0. At every
step, based on its current state, the automaton decides which counter
it should act upon and whether it should decrease it or increase
it. Then the action is \emph{attempted}. An attempt to decrease fails
iff the counter already contains~0; an attempt to increase fails iff
the counter is the primary one and it already contains~$n$; an attempt
to increase a secondary counter never fails. A failed attempt leaves
the counter contents intact; a successful attempt updates the counter
contents accordingly. Based on its current state and on whether the
attempt succeeded or not, the automaton selects a new state and moves
to it. The input is \emph{accepted} if the machine ever enters a
designated \emph{final state}, and the \emph{language} of the machine
is exactly the set of inputs that it accepts. If, for all~$n$, the
behavior of the automaton guarantees no secondary counter ever grows
larger than~$n$, we say the automaton is \emph{$($linearly$)$
bounded}.

We will be interested in a special version of the emptiness problem
for multi-counter automata. One way to introduce it is to start with
the emptiness problem for $\tm$s (``given a description of a $\tm$,
check that its language is empty''), which is well known to be
unrecognizable~\cite{houl79}, and to consider some `simplifications'
to~it:
\begin{nicelist}
\item
What happens if, instead of a full-fledged $\tm$, the machine we are
given is `simpler'? Say, a multi-counter automaton?  Or just a
$\dca_2$? Clearly, checking emptiness becomes `simpler', too. Does it
also become recognizable?
\item
What if, in addition, the given $\dca_2$ is promised to be
\emph{bounded}? And \emph{terminating}, too? And to also
\emph{obey a threshold}?  As the promise gets
stronger, checking emptiness again becomes `simpler'. But does it
become recognizable?
\end{nicelist}
So, the problem we want to define is: ``given a description of a
$\dca_2$ that is promised to be bounded and terminating and to obey a
threshold, check that its language is empty.'' In formal dialect,
$\eot=(\eoty,\eotn)$, where
$$ \eoty =  \{ z\in\desc{\dca_2'} \mid  \LLL{z}=\emptyset  \}
\qquad\text{and}\qquad \eotn =  \{ z\in\desc{\dca_2'} \mid
\LLL{z}\neq\emptyset  \},
$$ $\desc{\dca_2'}$ stands for the set of descriptions (under a fixed
encoding) of all terminating, bounded $\dca_2$s that obey a threshold,
and $\LLL{z}$ denotes the language of the machine described by
$z$. Interestingly, although not surprisingly, even for such a weak
automaton and under such a strong promise, emptiness remains
unrecognizable:\footnote{Note that clearly $\eot\in\Pi_1$ and that the
proof of Lemma~\ref{lem:main} will show $\eot$ is $\Pi_1$-complete. We
also remark that, under no promise and after non-trivially  modifying
the definition of $\dca_2$s, $\eot$ is the emptiness problem for
2-register machines, which is well known to be
$\Pi_1$-complete~\cite{mi61}.}

\begin{lemma} \label{lem:main}
$\eot$ is unrecognizable.
\end{lemma} 

\noindent
We use this fact in the next section, but defer proving it until
Section~\ref{sec:main-lemma}. In between,
Section~\ref{sec:programming-counters} discusses the capabilities of
multi-counter automata.

\section{The main theorem}
\label{sec:main-theorem}

We are now ready to state and prove the main theorem.

\begin{theorem}
Suppose $D$, $E$ are two \emph{comparable} descriptional systems that
satisfy conditions \textup{(H$_1$)--(H$_3$)} of
Lemma~\textup{\ref{lem:Hartmanis}}. If they also satisfy the following:
\begin{nicelist}
\item[\textup{(C$_1$)}]   there exists a name $e_0\in E$ that  has no
equivalent in $D$,
\item[\textup{(C$_2$)}]  given a description $z$ of a
\emph{terminating}, \emph{bounded} $\dca_2$ that \emph{obeys a
threshold}, we can effectively construct a name $e_z\in E$ such that
\begin{equation} \label{eq:e_z}
\LLL{e_z} \;=\; \LLL{e_0} \;\cup\; \{ w\in\ggG\s \mid
\textup{length}(w)\in\LLL{z} \},
\end{equation}
\item[\textup{(C$_3$)}] every \emph{co-finite} language has a name in
$D$ that maps to it,
\end{nicelist}
then the trade-off from $E$ to $D$ is non-recursive.
\end{theorem}

\noindent 
Before proving the theorem, we discuss how mild conditions
(C$_1$)--(C$_3$) really are.  Since every co-finite language is
regular, (C$_3$)~is trivially satisfied whenever the names in $D$
describe machines that have some kind of finite state control.

The second condition essentially says that the machines described by
$E$ have enough resources to simulate a bounded $\dca_2$. Because
then, from a given $z$, we can always construct the description $e_z$
of the $E$-machine that does the following:
\begin{quote}
on input $w\in\ggG\s$: first simulate on $\textup{length}(w)$ the
$\dca_2$ described by~$z$; if this accepts, then halt and accept;
otherwise, simulate on $w$ the machine described by $e_0$ and accept,
reject, or loop accordingly.
\end{quote}
and which obviously satisfies~(\ref{eq:e_z}) (note the importance of
the promise that $z$ describes a $\dca_2$ that never rejects by
looping).  Given how weak bounded $\dca_2$s are, most two-way machines
with non-regular capabilities will easily meet~(C$_2$).

The important condition is~(C$_1$), which requires that the machines
described by $D$ are \emph{not as powerful as} those described by $E$;
in other words, a separation is needed between the complexity classes
that correspond to the two systems.

\proof{We essentially repeat Hartmanis' argument
from~\cite[Example~4]{ha83} (see also~\cite[Theorem~7]{ku03}). Suppose
$D$, $E$ are as in the statement of the theorem. Since
\textup{(H$_1$)--(H$_3$)} are satisfied, Lemma~\ref{lem:Hartmanis}
implies that we only need to prove that the inadequacy problem $\II$
from $E$ to $D$ is unrecognizable. By Lemma~\ref{lem:main}, we just
need to reduce $\eot$ to~it:
$$ 
\eot \leq \II.
$$ 
Given a $z\in\desc{\dca_2'}$, we simply construct the name $e_z\in
E$ guaranteed by conditions~(C$_1$) and~(C$_2$), so that
$$ 
\LLL{e_z} \;=\; \LLL{e_0} \;\cup\; \{ w\in\ggG\s \mid
\textup{length}(w)\in\LLL{z} \}.
$$ 
If $z\in\eoty$, then the language of $z$ is empty, so that
$\LLL{e_z}=\LLL{e_0}$ and $e_z$ has no $D$-equivalent (because $e_0$
does not); hence $e_z\in\IIy$. On the other hand, if $z\in\eotn$, then
the language of $z$ contains all sufficiently large $w\in\ggG\s$, so
that $\LLL{e_z}$ is co-finite and has $D$-equivalents (by~(C$_3$));
hence $e_z\in\IIn$. This concludes the proof.\footnote{Note the
slightly stronger fact: $\II$ remains unrecognizable even under the
promise that the given $e\in E$ either has no $D$-equivalent or its
language is co-finite. In addition, the promise that the given
$\dca_2'$ obeys a threshold can be slightly relaxed: we only need to
know that its language is either empty or co-finite.}}

\section{Programming counters}
\label{sec:programming-counters}

In order to present the capabilities of multi-counter automata, we
introduce some `program' notation.  First, the two atomic operations,
the attempt to decrease a counter $X$ and the attempt to increase it,
are denoted respectively by
$$ X \stackrel{f}{\longleftarrow} X-1
\qquad\qquad\text{and}\qquad\qquad X \stackrel{f}{\longleftarrow} X+1,
$$ where, in each case, flag $f$ is set to true iff the attempt
succeeds. Then, the compound operation of setting $X$ to~$0$, denoted
by $X\longleftarrow 0$, can be described by 
\begin{equation} \label{eq:set-X-to-0}
\texttt{repeat }\; X \stackrel{f}{\longleftarrow} X-1\; \texttt{ until
$\neg f$}.
\end{equation}
If a second counter $Y$ is present, we can transfer the contents of
$Y$ into $X$: we set $X$ to~$0$, then repeatedly decrease $Y$ and
increase $X$ until $Y$ is~0. We denote this by
$$ (X,Y)\stackrel{f}{\longleftarrow}(Y,0),
$$ and describe it by a line similar to~(\ref{eq:set-X-to-0}). Note
that, if $X$ is the primary counter and $Y>n$, then one of the
attempts to increase $X$ will fail; in that case, we restore the
original value of $Y$ returning $X$ to~0, and set flag $f$ to
false. So, $X$'s original contents are always lost, but this never
happens to the original contents of~$Y$.

Changing how fast $X$ increases as $Y$ decreases, we can
multiply/divide $Y$ into $X$ by any constant $a\in\nats$. We denote
these operations by
$$ (X,Y) \stackrel{f}{\longleftarrow} (aY,0) \qquad\text{and}\qquad
(X,Y) \stackrel{\,f,r}{\longleftarrow}
\bigl(\bigl\lfloor\tfrac{Y}{a}\bigr\rfloor,0\bigr),
$$ where, in the second operation, we also find the remainder and
return it in~$r$.  As before, if $X$ is the primary counter and $aY>n$
(respectively, $\lfloor Y/a\rfloor>n$) then one of the attempts to
increase $X$ will fail; we then restore the original value of $Y$
returning $X$ to~0, and set flag $f$ to false.

At a higher level, we can attempt to multiply $Y$ by a constant $a$
(into $Y$) using $X$ as an auxiliary counter and making sure $Y$
changes only if the operation succeeds:\footnote{Note the use of
\texttt{t} in the place of a flag, indicating that the action is
guaranteed to succeed.}
$$ (X,Y) \stackrel{f}{\longleftarrow} (aY,0); \quad \texttt{if $f$
then } (Y,X) \stackrel{\texttt{t}}{\longleftarrow} (X,0).
$$ Division (with remainder) can be done similarly.  We denote the two
operations~by
\begin{equation} \label{eq:mult-div}
Y  \stackrel{\;\,f,X}{\longleftarrow} aY
\qquad\qquad\text{and}\qquad\qquad Y
\stackrel{\;f,r,X}{\longleftarrow}
\bigl\lfloor\tfrac{Y}{a}\bigr\rfloor.
\end{equation}
Now, if $X$ is primary, we can set $Y$ to the largest power of $a$
that can fit in $n$:
\begin{equation}  \label{eq:implement-log}
\begin{split}
& X\longleftarrow 0; \; X\stackrel{f}{\longleftarrow} X+1; \\ &
\texttt{if $f$ then } \{ Y\longleftarrow 0; \;
Y\stackrel{\texttt{t}}{\longleftarrow} Y+1; \; \texttt{repeat } Y
\stackrel{\;\,g,X}{\longleftarrow} aY \texttt{ until }\neg g \}
\end{split}
\end{equation}
(note that this fails iff $n=0$).  We denote~(\ref{eq:implement-log})
by (remember that $\lg_a n = \lfloor\log_a n\rfloor$):
$$ Y \stackrel{\;\,f,X}{\longleftarrow} a^{\lg_a n}.
$$ If a third counter $Z$ is present, we can
modify~(\ref{eq:implement-log}) to also count (in~$Z$) the number of
iterations performed. This gives us a way to calculate $\lg_a
n$ (failing iff $n=0$):
$$ Z  \stackrel{\;f,X,Y}{\longleftarrow}  \lg_a n.
$$ In another variation, we can modify the multiplication
in~(\ref{eq:mult-div}) so that the success of the operation depends on
the \emph{contents} of $X$ (as opposed to its upper bound~$n$):
$$ Y  \stackrel{\;\,f,X,Z}{\longleftarrow} aY,
$$ meaning that, using $Z$ as auxiliary and without affecting $X$: if
$aY\leq X$, then $Y$ is set to $aY$; otherwise, $Y$ is
unaffected.\footnote{To implement this, we first set $Z$ to~0. Then,
we repeatedly decrease $Y$, increase $Z$, and decrease $X$ by
$a$. \emph{If $X$ becomes~0 before $Y$}, then $aY>X$ and the operation
should fail: we restore the original values of $Y$ and $X$ by
repeatedly decreasing $Z$, increasing $Y$, and increasing $X$ by $a$,
until $Z$ becomes~0. \emph{Otherwise}, $aY\leq X$ and the operation
will succeed: we copy the correct value to $Y$ and restore the value
of $X$ by repeatedly decreasing $Z$ and increasing each of $Y$, $X$
by~$a$, until $Z$ becomes~0. Note that \emph{if originally $Y,Z\leq
X$, then at no point during the operation does any of the counters
assume a value greater than the original value
of~$X$}.\label{note:implement-bounded-mult}} Then, the following
variant of~(\ref{eq:implement-log})
\begin{multline*}
X\stackrel{f}{\longleftarrow} X-1; \;\; \texttt{if $f$ then } \{
X\stackrel{\texttt{t}}{\longleftarrow} X+1; \\ Y \longleftarrow 0;
\;\; Y\stackrel{\texttt{t}}{\longleftarrow} Y+1; \;\; \texttt{repeat }
Y \stackrel{\;\,g,X,Z}{\longleftarrow} aY \;\; \texttt{until $\neg g$}
\}
\end{multline*}
implements the attempt to set $Y$ to the largest power of $a$ that is
at most $X$, using $Z$ as auxiliary and leaving $X$ unaffected
(failing iff $X$ is~0). We denote this one by
$$ Y \stackrel{\;\,f,Z}{\longleftarrow} a^{\lg_a X}.
$$ It is important to note that, \emph{if originally $Y,Z\leq X$, then
during this operation no counter ever assumes a value greater than the
original value of $X$}
(cf. Footnote~\ref{note:implement-bounded-mult}).

\section{Proof of Lemma~\ref{lem:main}}
\label{sec:main-lemma}

We now prove that $\eot$ is unrecognizable. We do this by a reduction
from the complement of the halting problem, which is known to be
unrecognizable~\cite{houl79}:
$$ \looping \leq \eot,
$$ where $ \looping = \{  z\in\{\symz,\symo\}\s  \mid  \text{$z$
encodes a $\tm$ that loops on $z$}  \} $.  That is, we give an
algorithm that, on input a description $z$ of a $\tm$ $M$ produces a
description $z'$ of a terminating, bounded, threshold-obeying $\dca_2$
$M'$, such that
\begin{equation} \label{eq:reduction-condition-1}
M\text{ loops on }z  \implies \LLL{z'}=\emptyset
\qquad\text{and}\qquad M\text{ halts on }z  \implies
\LLL{z'}\neq\emptyset.
\end{equation}
In describing this algorithm, we will be calling a machine ($\tm$ or
$\dca_k$) \emph{good}, if it is terminating, bounded (for $\dca_k$s),
obeys a threshold, and its language
satisfies~(\ref{eq:reduction-condition-1}) when it replaces
$\lang{z'}$. Thus, for example, $M'$ will be good.
 
On its way to $z'$, the algorithm will construct descriptions of two
other machines: a description $z_A$ of a $\tm$ $A$, and a description
$z_B$ of a $\dca_3$ $B$.  In the sequence $M$, $A$, $B$, $M'$ each
machine after $M$ will be defined in terms of the previous one and
will be good. Our constructions use the ideas of~\cite{wa57}
and~\cite{mi61} (also found in~\cite{mi67,houl79}).

\medskip\noindent\emph{The first machine.}  $A$ is a $\tm$ with one
tape, infinite in both directions; the tape alphabet is
$\{\bla,\symz,\symo,\dotz,\doto\}$, while the input alphabet is
$\{\symz\}$. On input $\symz^n$, $A$ starts with tape contents
$$ \cdots\bla\bla\bla
\underbrace{\symz\symz\symz\cdots\symz\symz}_{n\text{ times}}
\bla\bla\bla\cdots
$$ and its head on the $\bla$ next to the leftmost~$\symz$ (or any
$\bla$, if $n=0$). It then computes:

\smallskip
\quad  1.~For all $w\in\{\symz,\symo\}^n$, from $\symz^n$ up to
$\symo^n$: \nopagebreak

\quad\qquad  --- if $w$ encodes a halting computation  history of~$M$
on~$z$, accept. \nopagebreak

\quad 2.~Reject.

\smallskip\noindent
The check inside the loop presupposes some fixed reasonable encoding
of sequences of configurations of $M$ into binary strings, with the
additional property: \emph{if $w$ encodes a computation history, then
every string of the form $w\symz^*$ encodes the same history}.

Note that, using the extra dotted symbols, $A$ can easily perform this
check without ever writing a non-blank symbol on a~$\bla$, or a $\bla$
on a non-blank symbol; and without ever visiting any $\bla$ that lies
beyond the two that originally delimit the input. As a consequence,
throughout its computation on $\symz^n$, $A$ keeps exactly $n$
non-blank symbols on its tape (occupying the same $n$ cells as the
symbols of the input). Also note that, by the selection of the
encoding scheme for $M$'s computation histories, if $A$ accepts
an input $\symz^n$, it necessarily accepts all longer inputs as well.

\medskip\noindent\emph{The second machine.}
$B$ is a $\dca_3$ that, on input $n\geq 30$, simulates the behavior of
$A$ on input $\symz^{\lg_5 \lg_{30} n}$; on input $n<30$, $B$ just
rejects.

To explain $B$'s behavior, let $J$, $L$, $R$ be its three
counters. $J$ is primary and helps performing operations on $L$ and
$R$, while $L$ and $R$ together encode tape configurations of $A$. To
see the encoding, consider the following example of a configuration:
\begin{center}
\smallskip
\footnotesize
$
\begin{array}{cccccccccccccc}
\cdots& l_4& l_3& l_2  & l_1  & l_0  &   h  & r_0  & r_1  & r_2  & r_3  & r_4 & r_5&\cdots \\ \hline
\cdots&\bla&\bla&\times&\times&\times&\times&\times&\times&\times&\times&\bla&\bla&\cdots \\ \hline
      &    &    &      &      &      &\uparrow&    &      &      &    &    &  
\end{array}
$
\smallskip
\end{center}
(here {\footnotesize$\times$} stands for any non-blank symbol,
{\footnotesize$\uparrow$} shows the head position). Mapping symbols
$\bla$, $\doto$, $\dotz$, $\symo$ and $\symz$ to numbers $0$, $1$,
$2$, $3$ and $4$, respectively,\footnote{Any mapping that maps symbol
$\bla$ to code~0 and symbol $\symz$ to code~4 will do.} we get each
tape cell map to a digit of the 5-ary numbering system. Then, the head
position splits the tape into three portions, which define the integers
$$
l = \sum_{i=0}^\infty l_i \cdot 5^i
\qquad\text{and}\qquad
h
\qquad\text{and}\qquad
r = \sum_{r=0}^\infty r_i \cdot 5^i
$$
(note that the two sums are finite exactly because $\bla$ maps
to~0). The values $l$ and $r$ are kept in $L$ and $R$,
while $h$ is kept in a register $H$ in $B$'s finite memory.

More specifically, on input $n$, $B$ starts with a two-part
initialization. First, it computes $\lg_{30} n$ into $J$, leaving $0$s
in $L$ and $R$ (this is if $n\geq 1$; if $n=0$, $B$ rejects):
\begin{equation*} 
R
\stackrel{\;\,f,J,L}{\longleftarrow}
\lg_{30} n;
\;\;
\texttt{if $\neg f$ then reject else }
\{
(J,R)\stackrel{\texttt{t}}{\longleftarrow} (R,0);
\;
L\longleftarrow 0
\}.
\end{equation*}
Then, it computes into $R$ the value $5^m-1$, where
$m=\lg_5\lg_{30}n$, leaving $0$s in $L$ and $H$ (this is only if
$J\geq 1$, that is if $n\geq 30$; otherwise, $n<30$ and $B$ rejects):
$$
R \stackrel{\;\,f,L}{\longleftarrow} 5^{\lg_5 J};
\;
\texttt{if $\neg f$ then reject else }
\{
R \stackrel{\texttt{t}}{\longleftarrow} R-1;
\;
L \longleftarrow 0;
\;
H \longleftarrow 0
\}.
$$
This completes the initialization, with $L=H=0$ and $R=5^m-1$, or in
5-ary:
$$
L = 0
\qquad\text{and}\qquad
H=0
\qquad\text{and}\qquad
R=
\underbrace{4\,\,4\,\,4\,\,\cdots\,\,4}_{m\text{ times}}.
$$
Hence $L$, $H$, $R$ correctly represent $A$'s starting tape
configuration on input $0^m$:
\begin{center}
\footnotesize
\smallskip
$
\begin{array}{cccccccccclc}
\cdots&l_1 &l_0 & h  & r_0 & r_1 & r_2 & \cdots &r_{m-1}&r_m &r_{m+1}&\cdots \\ \hline
\cdots&\bla&\bla&\bla&\symz&\symz&\symz& \cdots & \symz &\bla&\bla   &\cdots \\ \hline
      &    &    & \uparrow & 
\end{array}
$
\smallskip
\end{center}
(since symbol $\symz$ maps to~$4$) and $B$ is ready to start a
faithful simulation of $A$'s steps.

The automaton remembers in its finite memory the current state of $A$
as well as the code of the currently read symbol (in $H$). If $s$ is
the code of the new symbol to be written on the tape, $B$ computes
$$
L
\stackrel{\;\,\texttt{t},J}{\longleftarrow}
5L;
\;\;
\texttt{repeat $s$ times:}\;\;
L\stackrel{\texttt{t}}{\longleftarrow} L+1;
\;\;
R 
\stackrel{\texttt{t},r_0,J}{\longleftarrow}
\bigl\lfloor\tfrac{R}{5}\bigr\rfloor;
\;\;
H
\longleftarrow
r_0
$$
to simulate writing this symbol and moving to the right; it computes
$$
R
\stackrel{\;\,\texttt{t},J}{\longleftarrow}
5R;
\;\;
\texttt{repeat $s$ times:}\;\;
R\stackrel{\texttt{t}}{\longleftarrow} R+1;
\;\;
L 
\stackrel{\texttt{t},l_0,J}{\longleftarrow}
\bigl\lfloor\tfrac{L}{5}\bigr\rfloor;
\;\;
H
\longleftarrow
l_0
$$
to simulate writing this symbol and moving to the left.

It is important to note the range of the values assumed by the
counters. By the design of its main operation, the second part of the
initialization phase never assigns to a counter a value greater than
the original value of $J$, which is $\lg_{30} n$. Then, in the
simulation phase, the behavior of $A$ (the tape starts with $m$
$\symz$s and always contains exactly $m$ non-blank symbols) and the
selection of the symbol codes ($\symz$~gets the largest code) are such
that the initial value $5^m-1$ of $R$ upper bounds all possible values
that may appear in $B$'s counters. One consequence of this is that all
operations in the previous paragraph are guaranteed to be successful
(hence the $\texttt{t}$ reminder). Another consequence is that, since
$5^m-1<\lg_{30} n$, \emph{the entire computation of $B$ after the
first part of its initialization phase keeps all values of all
counters at or below} $\lg_{30} n$. This will prove crucial in the
next section.

\medskip\noindent\emph{The final machine.}
$M'$ is a $\dca_2$ that simulates the behavior of~$B$. If $U$, $V$ are
its two counters, then $U$ is primary and helps performing operations
on $V$, while $V$ encodes the contents of the counters of $B$:
whenever $J$, $L$, $R$, contain $j$, $l$, $r$ respectively, $V$
contains $2^j3^l5^r$.

The automaton starts by computing into $V$ the product
$30^t=2^t3^t5^t$, where $t=\lg_{30}n$ (this is only if $n\geq 1$; if
$n=0$, $M'$ rejects, exactly as $B$ would do):
$$
V \stackrel{\;\,f,U}{\longleftarrow} 30^{\lg_{30} n};
\;\;
\texttt{if $\neg f$ then reject}.
$$
It then removes all $3$s and $5$s from this product,\footnote{To
remove all $3$s, $M'$ divides $V$ by $3$ repeatedly 
($ 
V
\stackrel{\,\texttt{t},r,U}{\longleftarrow}
\bigl\lfloor\tfrac{V}{3}\bigr\rfloor 
$), 
until a non-zero remainder $r$ is returned, which implies
there were no $3$s in $V$ before the last division. Then the
correction 
$$
V 
\stackrel{\;\,\texttt{t},U}{\longleftarrow}
3V;
\;\;
\texttt{repeat $r$ times:}\;\;
V \stackrel{\texttt{t}}{\longleftarrow} V+1
$$
undoes the damage caused by the last division. A similar computation
removes all $5$s.} so that $V$ becomes $2^{\lg_{30}n}3^05^0$, which
correctly encodes the values of the counters of $B$ \emph{right after
the first part of its initialization phase}. Now $M'$ starts a faithful
simulation of the steps of~$B$.

The current state of $B$ is stored in $M'$'s memory. When $B$ tries
to decrease~$J$,
$$
J \stackrel{f}{\longleftarrow} J-1,
$$
$M'$ divides $V$ by~2; if there is no remainder, the division simulated
a successful decrement; otherwise, the simulated attempt failed
and $M'$ restores $V$'s initial value:
\begin{multline*}
V 
\stackrel{\,\texttt{t},r,U}{\longleftarrow}
\bigl\lfloor\tfrac{V}{2}\bigr\rfloor;
\;\;
\texttt{if $r=0$ then $f\longleftarrow \texttt{true}$ else}
\\
\{
f\longleftarrow \texttt{false};
\;\;
V 
\stackrel{\;\,\texttt{t},U}{\longleftarrow}
2V;
\;\;
\texttt{repeat $r$ times:}\;\;
V \stackrel{\texttt{t}}{\longleftarrow} V+1
\}.
\end{multline*}
The attempts to decrease $L$ or $R$ are handled similarly, with $3$ or
$5$ instead of~$2$.

Attempts of $B$ to increase its counters are of course simulated by
appropriate multiplications of $V$. The only subtlety involves failure
during increment attempts: how does the simulation make sure that
\emph{an attempt of $B$ to increase a counter fails} iff \emph{the
corresponding attempt of $M'$ to multiply $V$ fails?}\footnote{Note
that this condition is necessary for a successful simulation but is
not met in some obvious way. In $B$ the upper bound for $J$ is always
the same ($B$'s input), whereas in $M'$ the upper bound for its
representation is the base~2 logarithm of a value that depends (on
$M'$'s input and) on the values of the other two counters. Similarly,
in $B$ counter $L$ is unrestricted, whereas in $M'$ its representation
is bounded by a value that depends on the other two counters; 
similarly for~$R$.}  The crucial observation (from last section) is
that, since we are after the first part of $B$'s initialization phase,
\emph{no counter of $B$ ever assumes a value greater than $t=\lg_{30}
n$}. This immediately implies that, after the initialization phase of
$M'$, \emph{no counter of $M'$ ever assumes a value greater than
$2^t3^t5^t$}. Now, since $t<n$ and $2^t3^t5^t\leq n$, we conclude that
both (i)~all increment attempts of $B$ are successful, and (ii)~all
corresponding multiplication attempts of $M'$ are successful,
too. Hence, the equivalence above is satisfied, in a trivial way. Put
another way, when $M'$ multiplies $V$ to simulate a counter increment
in $B$, it knows in advance the increment is a successful one and the
multiplication will not fail. Overall, $B$'s atomic operation
$$
J \stackrel{\texttt{t}}{\longleftarrow} J+1
\qquad\text{is simulated by}\qquad
V \stackrel{\;\,\texttt{t},U}{\longleftarrow} 2V,
$$
and similarly for $L$ and $R$. 

As a final remark, we note the immediate by-product of our last
argument: Since $V$ clearly never exceeds $n$ during the
initialization phase of $M'$ and it also never exceeds $n$ during the
simulation of $B$, it follows that $M'$ is bounded.

\medskip\noindent   
This concludes the definitions of all three
machines in our reduction. It should be clear that $M'$ is good and
that a description $z'$ of it can be computed out of~$z$.

\section{Conclusion}
\label{sec:conclusion}

Using old ideas~\cite{wa57,mi61}, we showed the unrecognizability of
the emptiness problem for $\dca_2$s that are promised to be bounded,
always terminate, and obey a threshold. We then combined this with the
idea of~\cite{haba79a} to show that, if
machines~$\modelofcomputation{a}$ have the resources to simulate
$\dca_2$s of the particular kind and can also solve problems that
machines~$\modelofcomputation{b}$ cannot, then typically the trade-off
from~$\modelofcomputation{a}$ to~$\modelofcomputation{b}$ is
non-recursive.  Applying the theorem, we derived such trade-offs in
many conversions.

We do not know if the emptiness problem of
Section~\ref{sec:multicounter-automata} remains unrecognizable even
when the underlying machine is a \emph{$2$-register
automaton}~\cite{mi61} (that is, a $\dca_2$ that starts with $n$ in
its primary counter and where increments of that counter never
fail). If it is, then our main theorem can be made slightly stronger.  

\subsection*{Acknowledgments} 
I would like to thank Oscar Ibarra and Albert Meyer for discussions
that revealed the generality behind the argument of~\cite{ka04b}.

\end{document}